\documentclass[aps,twocolumn,superscriptaddress,nofootinbib,a4paper,longbibliography]{revtex4-2}
\usepackage{amsmath}
\usepackage{amsfonts}
\usepackage{amsthm}
\usepackage{amssymb}
\usepackage{graphicx}
\usepackage{enumerate}
\usepackage{color}
\usepackage[T1]{fontenc}
\usepackage{braket}
\usepackage{todonotes}
\usepackage{soul}
\usepackage{mathrsfs}
\usepackage{float}
\usepackage{bm}
\usepackage{multirow}
\usepackage{tabularx}
\usepackage{outlines}

\graphicspath{{figures/}}

\newtheorem{theorem}{Theorem}
\newtheorem{lemma}{Lemma}

\newtheorem*{theorem*}{Theorem}
\newtheorem{definition}{Definition}

\newtheorem{assumption}{Assumption}

\def\>{\rangle}
\def\<{\langle}

\usepackage[bookmarks=false,colorlinks,citecolor=blue,
linkcolor=blue,anchorcolor=blue,urlcolor=blue
]{hyperref}

\newcommand{\R}{\mathbb{R}}
\newcommand{\C}{\mathbb{C}}
\newcommand{\QA}{\mathcal{A}}
\newcommand{\QY}{\mathcal{A}_{Y}}
\newcommand{\Qset}{\mathcal{Q}_{2222}}

\newcommand{\Qsharp}{\mathcal{Q}_{2\mathrm{Q}}^{\mathrm{sharp}}}

\begin{document}
\title{No-Disturbance-without-Uncertainty generates the quantum set in the simplest Bell scenario}
\author{Ravishankar Ramanathan}
\email{ravi@cs.hku.hk}
\affiliation{School of Computing and Data Science, The University of Hong Kong, Pokfulam Road, Hong Kong}

\begin{abstract}
Characterizing the set of quantum nonlocal correlations through information-theoretic constraints is a central problem in quantum foundations.
We show that the convex hull of the no-signalling behaviors in a bipartite Bell experiment satisfying a common-overlap no-disturbance-without-uncertainty condition is exactly the finite-dimensional quantum set whenever one party has two binary measurements, while the other party may have arbitrary finite numbers of settings and outcomes. In the simplest binary input-output Bell scenario, this characterizes the full eight-dimensional probability set, including the local marginals. The proof constructs a quantum realization from the conditional probabilities and no-signalling. 
These results highlight the connection between local uncertainty-disturbance constraints and quantum nonlocality, 
mediated by no-signalling and shared classical randomness.  

\end{abstract}

\maketitle

\textit{Introduction.---}
The outcomes of local measurements on entangled quantum systems can display correlations that no local hidden-variable model can reproduce~\cite{Bell1964,CHSH1969}.
This phenomenon of Bell nonlocality is a fundamental feature of quantum theory, showcasing its departure from classical theories~\cite{PR94,Tsirelson1987,Popescu14}.
Besides its foundational interest, Bell nonlocality also underlies device-independent information processing.
Quantum key distribution, randomness generation, and randomness amplification can be certified from observed input--output statistics without a detailed characterization of the devices, under the causal and input assumptions of the protocol~\cite{Acin2007,Pironio2010,ColbeckRenner2012,VaziraniVidick2014,Brandao2016,Ramanathan2025bound}.
Characterizing the set of quantum nonlocal behaviors is therefore important both for the limits of nonlocality and for the capabilities of those protocols.

Information-theoretic principles have been sought that constrain the behaviors in Bell experiments directly, in a device-independent manner.
This program began with the observation that the no-signalling principle alone allows stronger correlations than quantum theory~\cite{PR94,Tsirelson1987}.
It was found that access to the so-called Popescu--Rohrlich boxes would collapse communication complexity~\cite{vanDam2013}.
A sequence of principles was then proposed to cut the no-signalling set from the outside: Information Causality~\cite{Pawlowski2009}, Macroscopic Locality~\cite{NavascuesWunderlich2010}, Local Orthogonality~\cite{Fritz2013} and others \cite{Navascues2015}; these principles produce constraints based on information-theoretic tasks.
A closer conceptual setting for the present work is a partly information-theoretic principle based on a combination of fine-grained uncertainty relations with steering~\cite{OppenheimWehner2010}, i.e., local uncertainty restricts the conditional statistics of one player, while steering determines which ensembles are compatible with a common bipartite state.
In the simplest two-input binary-output Bell scenario, the four-dimensional quantum correlation set (the set of correlators excluding the marginals) admits an elegant characterization due to Tsirelson-Landau-Masanes (TLM) ~\cite{Tsirelson1980,Landau1988,LandauEmpirical1988,Masanes2003}. Quantum Bell inequalities and the geometry of the quantum set have since been studied intensively ~\cite{RTHH16,  Goh2018, Liu2024, BarizienBancal2024}. The state-of-the-art result in the information-theoretic reconstruction of the quantum boundary is that a generalized notion of Information Causality recovers the exact TLM correlator boundary for the four-dimensional set and also yields constraints involving local marginals~\cite{GachechiladzeMiklin2026}.
On the other hand, the full eight-dimensional quantum set, including the local marginals, is known through its extremal points~\cite{BarizienBancal2025}, while thus far no principle (partly or wholly information-theoretic) is known to fully characterize this set.


A novel approach in the program is a principle termed No-Disturbance-without-Uncertainty (NDWU) formulated by Sun, Zhou and Yu in \cite{SunYu2019}, see also \cite{ZSY24, SZZW+26, SYC18,SSYC22, Hsu25}.
This principle states that the disturbance caused by one sharp local measurement on a subsequent (incompatible) measurement is bounded by a product of two uncertainties, the initial uncertainty of the observable and a second uncertainty depending on the pair of measurements.
For the case of two binary measurements, this NDWU principle was shown to be a single quadratic constraint involving one sequential overlap parameter $c$ for the two measurements and two individual outcome averages. In the simplest two-input binary-output Bell scenario, when applied to the four states that a distant measurement by Alice can prepare on Bob's side, this NDWU principle yields a criterion on the allowed behaviors.
In particular Ref.~\cite{SunYu2019} derived Tsirelson's bound for the CHSH inequality \cite{CHSH1969}, obtained constraints on a three-parameter noisy-box family, and recovered quantum boundaries on particular slices in this Bell scenario.
They left the question open whether every behavior passing the criterion admits a quantum realization.
We prove that implication and identify the quantum set with the convex hull of the behaviors satisfying the NDWU principle.


The paper is organized as follows. We first state the NDWU principle, establish the nonconvexity of the associated sets, and then  prove constructively that the convex hull of behaviors obeying this criterion on one side generates the eight-dimensional quantum set $\mathcal{Q}_{2222}$. 
As a consequence of this identification, the quantum value of any linear Bell functional on $\mathcal{Q}_{2222}$ is given exactly by a one-parameter family of second-order cone programs, one for each overlap. The same exactness is available from steering considerations; the present formulation makes that reduction explicit as a single family of cones. We then consider a slightly extended Bell scenario in which one player has two binary measurements while the other player has arbitrary finite numbers of settings and outcomes and use a recent result \cite{PanahiWolfe2025} to show how the convex hull of behaviors satisfying NDWU on one side also generates the quantum set in this scenario. We conclude by briefly discussing the extensions of the principle to scenarios with more outcomes while leaving this question open for forthcoming work. 


\textit{The No-Disturbance-without-Uncertainty principle.---}
Consider a bipartite Bell experiment in which Alice and Bob each choose one of $n$ measurement settings and each obtain one of $m$ outcomes. Let us denote $x\in X$ and $y\in Y$ for the settings and $a\in A$ and $b\in B$ for the outcomes, with $|X|=|Y|=n$ and $|A|=|B|=m$. The parties observe outcomes given inputs according to a  family of conditional probabilities $\{P(a,b\mid x,y)\}$ termed a behavior in the Bell scenario $(n,m;n,m)$.
The No-Signalling set $\mathcal{NS}(n,m;n,m)$ in this Bell scenario consists of those behaviors that are nonnegative and normalized, i.e., $\sum_{a,b}P(a,b\mid x,y)=1$ for every $x,y$, and that obey for all $a,x,y$
$\sum_b P(a,b\mid x,y)
=
P(a\mid x)$, and similarly for all $b,y,x$
$\sum_a P(a,b\mid x,y)
=
P(b\mid y)$.
The local set $\mathcal{L}(n,m;n,m)$ in this Bell scenario consists of those behaviors that admit a local hidden-variable decomposition. That is, there is a hidden variable $\lambda$, probabilities $q(\lambda)\ge 0$ with $\sum_\lambda q(\lambda)=1$, and response functions $P_A(a\mid x,\lambda)$ and $P_B(b\mid y,\lambda)$, such that
$P(a,b\mid x,y)
=
\sum_\lambda q(\lambda)\,P_A(a\mid x,\lambda)\,P_B(b\mid y,\lambda)$.
Both $\mathcal{NS}(n,m;n,m)$ and $\mathcal L(n,m;n,m)$ are convex polytopes, while the quantum set $\mathcal Q(n,m;n,m)$ lies between them.
This latter set consists of those behaviors for which there exist Hilbert spaces $\mathcal H_A,\mathcal H_B$, a state $\rho_{AB}$ on $\mathcal H_A\otimes\mathcal H_B$, and POVMs $\{M_{a|x}\}$, $\{N_{b|y}\}$ such that
$P(a,b\mid x,y)=\operatorname{Tr}\bigl[\rho_{AB}\,(M_{a|x}\otimes N_{b|y})\bigr]$.
One has the standard inclusions \cite{Brunner2014}
\[
\mathcal L(n,m;n,m)
\subseteq
\mathcal Q(n,m;n,m)
\subseteq
\mathcal{NS}(n,m;n,m).
\]

We begin with the minimal scenario $n=m=2$ which is the main focus of this work.
For notational convenience we write $\Qset$ for $\mathcal Q(2,2;2,2)$, noting that in this paper $\Qset$ denotes the finite-dimensional tensor-product quantum set. The inputs are $x,y\in\{0,1\}$ and the outputs are $a,b\in\{+1,-1\}$. 
A no-signalling behavior is labelled by the eight parameters for $x, y \in \{0,1\}$ given as
\[
\alpha_x=\mathbb E[A_x],\qquad
\beta_y=\mathbb E[B_y],\qquad
E_{xy}=\mathbb E[A_x B_y],
\]
with
\begin{equation}
P(a,b\mid x,y)
=\frac14\bigl(1+a\alpha_x+b\beta_y+ab E_{xy}\bigr).
\label{eq:ns}
\end{equation}
Jordan's lemma states that any two orthogonal projections on a finite-dimensional Hilbert space can be block-diagonalized simultaneously into blocks of dimension at most two~\cite{Jordan1875}.
Applied to two binary projective measurements, this decomposes one party's Hilbert space into one- and two-dimensional blocks.
In seminal work, Masanes showed that every extreme point of $\Qset$ is realized by a pure two-qubit state and one projective measurement per setting on each party's system~\cite{Masanes2005}.
Let $\Qsharp$ denote the (non-convex) set of $(2,2;2,2)$ behaviors $\{P(a,b\mid x,y)\}$ for which there exist a pure state $\lvert\psi\rangle\in\C^2\otimes\C^2$ and unit vectors $\hat{m}_x, \hat{n}_y\in\R^3$ such that
\begin{equation}
P(a,b\mid x,y)
=
\bigl\langle\psi\bigr\rvert
\frac{I+a\, \hat{m}_x\cdot \vec{\sigma}}{2}
\otimes
\frac{I+b\, \hat{n}_y\cdot \vec{\sigma}}{2}
\bigl\lvert\psi\bigr\rangle,
\label{eq:qsharp}
\end{equation}
where $\vec{\sigma} = (\sigma_x, \sigma_y, \sigma_z)$ denotes the three Pauli matrices.
Masanes' theorem is then the inclusion $\operatorname{ext}(\Qset)\subseteq\Qsharp$, so that we have
\[
\operatorname{ext}(\Qset)
\subseteq
\Qsharp
\subseteq
\Qset
=
\operatorname{conv}\bigl(\operatorname{ext}(\Qset)\bigr),
\]
where $\operatorname{ext}(\Qset)$ denotes the set of extreme points of $\Qset$ and $\operatorname{conv}\bigl(\operatorname{ext}(\Qset)\bigr)$ denotes their convex hull. 
We emphasize that $\Qset$ is the eight-dimensional convex body of all quantum behaviors in the Bell scenario $(2,2;2,2)$, including the local marginals, not just the four-dimensional quantum correlator set~\cite{Tsirelson1980,Landau1988,Masanes2003,Goh2018,BarizienBancal2025}.

We now proceed to the NDWU principle and its physical significance~\cite{SunYu2019}.
Let us consider Bob's two binary measurements $B_0$ and $B_1$, with outcomes labeled $\pm 1$.
The principle concerns four possible experiments: Bob may measure $B_0$ alone, Bob may measure $B_1$ alone, Bob may measure $B_0$ first followed by $B_1$, or Bob may measure $B_1$ first followed by $B_0$. A first assumption made in the principle is that the measurements are sharp, meaning that if either measurement is performed twice in succession, the second outcome equals the first with probability one. Secondly and more importantly, we assume repeatability and outcome-determined post-measurement states: after the outcome of one measurement, say $B_0$, has been obtained, the post-measurement state depends only on that outcome (hence the probabilities for a subsequent $B_1$ measurement depend only on that outcome), and in particular they do not depend on how Bob's system was prepared prior to the $B_0$ measurement (we write the assumptions formally in Appendix ~\ref{app:ndwu}). These assumptions together with the quantitative version of the principle applied to the two possible sequences and every preparation of Bob's system, imply that the two post-measurement averages are opposite, i.e.,  
\begin{equation}
\mathbb E[B_1\mid B_0=+1] = c = -\mathbb E[B_1\mid B_0=-1],
\label{eq:c}
\end{equation}
with a parameter \(c\in[-1,1]\) that is a property of the pair $(B_0,B_1)$ and not of any initial prepared state, the derivation is detailed in  Lemma~\ref{lem:opposite} of Appendix ~\ref{app:ndwu}. 



Now consider any prepared state of Bob's system, and let us denote by $v_y$ the average of $B_y$ when that measurement is performed alone for $y \in \{0,1\}$.
An experiment that measures $B_1$ alone has expectation $v_1$.
A sequential experiment, which measures $B_0$ first and then $B_1$, does not in general return the same value.
The outcome of $B_0$ is $+1$ with probability $(1+v_0)/2$ and $-1$ with probability $(1-v_0)/2$.
After those two outcomes, the averages of $B_1$ are $c$ and $-c$ respectively by \eqref{eq:c}, so that the average of $B_1$ at the end of the sequential experiment is
$\mathbb E_{\mathrm{seq}}[B_1]
=
\frac{1+v_0}{2}\,c+\frac{1-v_0}{2}\,(-c)
=
c v_0$. The difference between measuring $B_1$ alone and measuring $B_1$ after first measuring $B_0$ is 
\begin{equation}
D_{B_0\to B_1}
:=
\bigl|v_1-c v_0\bigr|,
\label{eq:D}
\end{equation}
and is termed the disturbance of $B_1$ caused by first measuring $B_0$. Note that if $v_0=\pm 1$ (no uncertainty), the NDWU principle requires the disturbance $D_{B_0\to B_1}$ to vanish and hence \(v_1=\pm c\), that is a measurement whose outcome is certain does not disturb $B_1$.
On the other hand, if $v_0$ is closer to zero, the outcome of $B_0$ is not known in advance, and first measuring $B_0$ may change the average of $B_1$.

The outcome \(B_0=\pm 1\) has average \(v_0\), and \(B_0^2=1\) so that the standard deviation of $B_0$ is
\begin{equation}
\Delta_{B_0} = \sqrt{\mathbb E[B_0^2]-\mathbb E[B_0]^2}
=
\sqrt{1-v_0^2}.
\label{eq:Delta}
\end{equation}
Furthermore, by Lemma \ref{lem:opposite} of Appendix \ref{app:ndwu}, under the stated assumptions it holds that $\mathbb{E}[B_0 \mid B_1 = +1] = c = - \mathbb{E}[B_0 \mid B_1 = -1]$. Thus,
the standard deviation of $B_0$ on either state left after measuring $B_1$ is $\sqrt{1-c^2}$. Denote
\begin{equation}
\Delta_{B_0\mid B_1}
:= \frac12 \sum_{b = \pm 1} \sqrt{1-\mathbb{E}[B_0 \mid B_1=b]^2} 
=
\sqrt{1-c^2}.
\label{eq:Deltac}
\end{equation}
Once Lemma~\ref{lem:opposite} establishes that the two branch averages are opposite, any convex weighting yields the same value.
No-Disturbance-without-Uncertainty is then the statement that the disturbance cannot exceed the product of these two standard deviations~\cite{SunYu2019}, i.e.,
\begin{equation}
\Delta_{B_0}\,\Delta_{B_0\mid B_1}
\ge
D_{B_0\to B_1}.
\label{eq:ndwu-prod}
\end{equation}
Squaring both sides of~\eqref{eq:ndwu-prod}, which are nonnegative, gives
\begin{equation}
v_0^2+v_1^2+c^2-2c\,v_0 v_1
\le 1.
\label{eq:ndwu-bin}
\end{equation}
Note that the number \(c\) is fixed by the pair \((B_0,B_1)\), while the pair \((v_0,v_1)\) changes with the state prepared by Alice's measurement \(x\) and outcome \(a\). 

In the Bell experiment Alice's setting $x$ and outcome $a$ prepare Bob's system.
The single-measurement averages on that preparation are the conditional expectations given by (see Appendix \ref{app:ndwu})
\begin{equation}
v_y^{a|x}
=
\mathbb E[B_y\mid A_x=a]
=
\frac{\beta_y+a E_{xy}}{1+a\alpha_x},
\label{eq:EB}
\end{equation}
defined for all $(a,x)$ with $P(a\mid x)>0$.
For $(a,x)$ with $P(a\mid x)=0$ we set $v_y^{a|x}=0$.
NDWU on Bob's side requires one and the same $c_B$ for every such preparation (this reflects the condition that $c_B$ depends only on $(B_0,B_1)$ and not on the preparation).
Analogously we have
\begin{equation}
u_x^{b|y}
=
\mathbb E[A_x\mid B_y=b]
=
\frac{\alpha_x+b E_{xy}}{1+b\beta_y},
\label{eq:EA}
\end{equation}
and we set $u_x^{b|y}=0$ whenever $P(b\mid y)=0$.

We now proceed to define the set of behaviors satisfying the NDWU principle on either side as $\QA_A$, $\QA_B$, and on both sides as $\QY$. We also refer to the constraints \eqref{eq:B} and \eqref{eq:A} as the common-overlap NDWU criterion.
\begin{definition}
\label{def:sets}
A no-signalling $(2,2;2,2)$ behavior $\{P(a,b\mid x,y)\}$ belongs to the set $\QA_B$ if there exists $c_B\in[-1,1]$ (which may depend on $P$ but not on $(x,a)$) such that
\begin{equation}
\bigl(v_0^{a|x}\bigr)^2
+\bigl(v_1^{a|x}\bigr)^2
+c_B^2
-2c_B\,v_0^{a|x}v_1^{a|x}
\le 1
\label{eq:B}
\end{equation}
holds for every $(x,a)$ with $P(a\mid x)>0$.
It belongs to the set $\QA_A$ if there exists $c_A\in[-1,1]$ (which may depend on $P$ but not on $(y,b)$) such that
\begin{equation}
\bigl(u_0^{b|y}\bigr)^2
+\bigl(u_1^{b|y}\bigr)^2
+c_A^2
-2c_A\,u_0^{b|y}u_1^{b|y}
\le 1
\label{eq:A}
\end{equation}
holds for every $(y,b)$ with $P(b\mid y)>0$.
The set of behaviors that satisfy NDWU on both parties is
\[
\QY
:=
\QA_A\cap\QA_B.
\]
\end{definition}
Note that the sets defined here depend only on the observed probabilities in the Bell experiment and the existence of a common overlap parameter; the sequential measurement assumptions only provide their operational motivation.  

For any given $(v_0,v_1)\in[-1,1]^2$ the values of $c$ compatible with the NDWU condition~\eqref{eq:ndwu-bin} form a closed interval $I(v_0,v_1)$ given as (see the derivation in Appendix~\ref{app:ndwu})
\begin{equation}
\begin{aligned}
I(v_0,v_1)
:=
\Bigl[&
v_0 v_1-\sqrt{(1-v_0^2)(1-v_1^2)},\\
&
v_0 v_1+\sqrt{(1-v_0^2)(1-v_1^2)}
\Bigr].
\end{aligned}
\label{eq:Iv}
\end{equation}
Hence $P\in\QA_B$ if and only if
\begin{equation}
\bigcap_{\substack{x,a\\ P(a\mid x)>0}}
I\bigl(v_0^{a|x},v_1^{a|x}\bigr)
\neq\varnothing,
\label{eq:intersectB}
\end{equation}
and analogously $P\in\QA_A$ if and only if the corresponding intersection over the pairs~\eqref{eq:EA} is non-empty.

\textit{conv$(\mathcal A_A)=$conv$(\mathcal A_B)=$conv$(\mathcal A_Y)=\mathcal Q_{2222}$.---}
The sets just defined are constraints on a single behavior $\{P(a,b|x,y)\}$.
However they are not convex, and more interestingly not every local (hence quantum) behavior even obeys them (this property was also noticed by Sun et al. \cite{SunYu2019}); an explicit local point outside $\QY$ is given in Appendix~\ref{app:nonconvex}. While every local deterministic behavior belongs to $\QY$, mixtures of local deterministic behaviors may use a different overlap parameter for each deterministic term, while the NDWU principle requires a single overlap parameter for the whole observed behavior, in particular  mixing local deterministic behaviors that require opposite overlaps $c=\pm 1$ can produce a local behavior admitting no common overlap.
Nevertheless we prove that the convex hull of $\QY$ (and in fact already of $\QA_A$ and $\QA_B$ separately) constitutes the quantum set $\Qset$.

\begin{theorem}
\label{thm:hull}
It holds that
\[
\operatorname{conv}(\QA_A)
=
\operatorname{conv}(\QA_B)
=
\operatorname{conv}(\QY)
=
\Qset.
\]
Furthermore, every point of $\Qset$ is a convex combination of at most eight points of $\Qsharp$, and therefore of $\QY$.
\end{theorem}

\textit{Sketch of the proof.---}
Since one party, say Bob, has two binary measurements, Jordan's lemma decomposes any finite-dimensional realization into components with a qubit on that side. On a qubit the common-overlap NDWU quadratic is exactly the condition that the two conditional averages \((v_0^{a|x},v_1^{a|x})\) live inside the elliptical region obtained by projecting the Bloch ball onto the plane of two projective measurements whose overlap is \(c_B\). No-signalling forces the positive operators \(\tau_{a|x}\) constructed from those averages to sum to the same state for both of Alice's settings. A standard purification argument using the Schr\"{o}dinger-HJW theorem \cite{Hughston1993, Schrodinger1936} then converts this collection of operators into a bipartite quantum state and local measurements that reproduce the original probabilities.

More explicitly, to prove the constructive inclusion
$\mathcal A_B\subseteq\mathcal Q_{2222}$, take a no-signalling
behavior admitting a common overlap $c_B$, with $|c_B|<1$, and set
$s_B=\sqrt{1-c_B^2}$.
Choose Bob's observables as
$B_0=\sigma_z$ and $B_1=s_B\sigma_x+c_B\sigma_z$.
For every branch with $P(a|x) > 0$, define
\begin{equation}
\begin{aligned}
\tau_{a|x}
=\frac{P(a\mid x)}{2}\Biggl[
I
+\frac{v_1^{a|x}-c_Bv_0^{a|x}}{s_B}\,\sigma_x
+v_0^{a|x}\sigma_z
\Biggr],
\end{aligned}
\label{eq:constructive-lift-main}
\end{equation}
and set $\tau_{a|x}=0$ when $P(a|x) = 0$.
The NDWU quadratic is exactly the condition
$\tau_{a|x}\succeq0$.
Crucially, the lift is linear in the unnormalized conditional
statistics $P(a\mid x)v_y^{a|x}$.
No-signalling therefore gives
\begin{equation}
\begin{aligned}
\sum_a\tau_{a|x}
=\rho_B
=\frac12\left[
I+\frac{\beta_1-c_B\beta_0}{s_B}\,\sigma_x
+\beta_0\sigma_z
\right],
\end{aligned}
\label{eq:common-average-main}
\end{equation}
independently of Alice's setting $x$.
Thus the local uncertainty--disturbance constraint supplies
positive operators, while no-signalling supplies their
compatibility as decompositions of one state.

Purifying $\rho_B$, the Schr\"odinger-HJW construction provides
Alice's POVMs $\{M_{a|x}\}_a$ satisfying
\begin{equation}
\begin{aligned}
\langle\Psi|M_{a|x}\otimes N_{b|y}|\Psi\rangle
&=\operatorname{Tr}(\tau_{a|x}N_{b|y})\\
&=P(a,b\mid x,y),
\end{aligned}
\label{eq:reproduction-main}
\end{equation}
where $N_{b|y}=(I+bB_y)/2$.
At $c_B=\pm1$, the criterion forces
$v_1^{a|x}=c_Bv_0^{a|x}$; the same construction uses
$B_1=c_BB_0$ and
$\tau_{a|x}=P(a\mid x)(I+v_0^{a|x}\sigma_z)/2$.

For the reverse inclusion, projective dilation and Jordan
decomposition on both parties express every quantum behavior
as a finite convex combination of pure two-qubit projective
behaviors.
Each such component obeys the common-overlap criterion on
both sides, with overlaps that may differ between components.
Consequently
$\mathcal Q_{2222}
=\operatorname{conv}(\mathcal Q_{2\mathrm Q}^{\mathrm{sharp}})
\subseteq\operatorname{conv}(\mathcal A_Y)$,
which completes the proof. $\blacksquare$

The full proof of Thm. \ref{thm:hull} is given in Appendix~\ref{app:hierarchy}. In Appendix~\ref{app:sharp}, we investigate the relationship of these sets with $\Qsharp$ and show that $\Qsharp$ is strictly smaller than $\QY$.

\textit{A second-order cone model of the identification.---}
The following records the concrete model that the identification $\operatorname{conv}(\mathcal A_B)=\mathcal Q_{2222}$ supplies. Consider a linear Bell expression $I(P)$ in this scenario given as
\[
I(P)=\sum_{x,y,a,b} I_{xy}^{ab}P(a,b\mid x,y),
\]
for some real coefficients $I_{xy}^{ab}$.
For a fixed $c\in[-1,1]$, let us set $s=\sqrt{1-c^2}$ and for each $(a,x)$, define
\[
t_{a|x}=1+a\alpha_x,\quad
p_{a|x}=\beta_0+aE_{x0},\quad
q_{a|x}=\beta_1+aE_{x1}.
\]
Then the common-overlap condition \eqref{eq:B} can be written on no-signalling behaviors as
\begin{equation}
\left\|
\begin{pmatrix}
q_{a|x}-c\,p_{a|x}\\
s\,p_{a|x}
\end{pmatrix}
\right\|_2
\le s\,t_{a|x}
\qquad\text{for every }(a,x).
\label{eq:socp-cone-main}
\end{equation}
This is seen by simply observing that the conditional averages are $v_{0}^{a|x} = p_{a|x}/t_{a|x}$ and $v_{1}^{a|x} = q_{a|x}/t_{a|x}$. 
And for $|c|<1$, $(q_{a|x} - c p_{a|x})^2 + (s p_{a|x})^2 \leq (1 - c^2) t_{a|x}^2$ is just \eqref{eq:B} multiplied by $t_{a|x}^2$.
At $c=\pm 1$, it gives $q_{a|x}=c\,p_{a|x}$. At $c=\pm 1$, the program reduces to a linear program, with the usual probability nonnegativity constraints retained. In the case when $t_{a|x} = 0$, positivity of the probabilities implies $p_{a|x} = q_{a|x} = 0$ so \eqref{eq:socp-cone-main} also covers this case. 

Let $\Phi_I^B(c)$ be the maximum of $I(P)$ over no-signalling behaviors satisfying~\eqref{eq:socp-cone-main}.
For fixed $c$, this is a second-order cone program.
Since a linear functional has the same supremum on a set and its convex hull, we obtain that the maximum quantum value for the Bell inequality can be expressed as
\begin{equation}
\omega_Q(I)
=
\max_{c\in[-1,1]}\Phi_I^B(c).
\label{eq:socp-quantum-value}
\end{equation}
The maximum is attained because the joint feasible set of $(P,c)$ is compact.
Once the set is identified, the quantum maximum of any linear functional is the maximum over the overlap parameter of a second-order cone program on the corresponding fibre. This family is the concrete model that the identification supplies: it follows from decomposing the set according to the overlap parameter, together with the fact that each qubit set at fixed overlap is a second-order cone. The same exactness is available from steering considerations; the present formulation makes that reduction explicit as a single family of cones. At each fixed overlap, only four second-order cone constraints are required, one for each conditional preparation on Bob's side.


\textit{One party with two binary measurements.---}
The construction extends directly when only one party has two binary measurements, using the qubit reduction discussed in \cite{PanahiWolfe2025}.
Consider the Bell scenario in which Alice has settings $x\in\{0,1\}$ and outcomes $a\in\{\pm 1\}$, while Bob has $n$ settings and an arbitrary finite number $m_y$  of outcomes for setting $y$.
Let us write $\mathbf m=(m_y)_y$ and denote the quantum set by $\mathcal Q(2,2;n,\mathbf m)$.
Let $\mathcal A_A^{(n,\mathbf m)}$ be the set of no-signalling behaviors admitting one $c_A\in[-1,1]$ on Alice's side such that
\begin{equation}
\bigl(u_0^{b|y}\bigr)^2+\bigl(u_1^{b|y}\bigr)^2+c_A^2-2c_A u_0^{b|y}u_1^{b|y}\le 1
\label{eq:general-one-sided-ndwu}
\end{equation}
for every $(b,y)$ with $P(b\mid y)>0$, where $u_x^{b|y}=\mathbb E[A_x\mid B_y=b]$.
Then it holds that
\begin{equation}
\mathcal Q(2,2;n,\mathbf m)
=
\operatorname{conv}\bigl(\mathcal A_A^{(n,\mathbf m)}\bigr).
\label{eq:general-one-sided-hull}
\end{equation}
The proof follows analogously to that of Thm. \ref{thm:hull} and is given in Appendix~\ref{app:beyond}.

\textit{Measurements with more outcomes.---}
When both parties have two measurements with $m$ outcomes, the binary result gives a necessary condition by coarse-graining.
For each setting, choose nonempty proper subsets $S_x,T_y\subset\{1,\ldots,m\}$ and record whether $a\in S_x$ and $b\in T_y$.
Denote the resulting binary behavior by $P_{S,T}$, and define
\[
\mathcal R_m^{\mathrm{bin}}
=
\left\{
P\in\mathcal{NS}(2,m;2,m):
P_{S,T}\in\operatorname{conv}(\mathcal A_Y)
\text{ for all }S,T
\right\}.
\]
Coarse-graining a quantum measurement amounts to summing measurement operators.
Every quantum behavior therefore satisfies
\begin{equation}
\mathcal Q(2,m;2,m)\subseteq\mathcal R_m^{\mathrm{bin}}.
\label{eq:binning-outer-main}
\end{equation}
For more settings, the same test applies to every selected pair of settings on each side (see details in  Appendix~\ref{app:beyond}).
This is an outer approximation obtained from the binary theorem, an intrinsic extension of the principle to more inputs and outputs is left for future work.

\textit{Conclusions.---}
The relation between local uncertainty constraints and the set of nonlocal correlations is a central question in the search for physical principles underlying quantum theory. We have shown that the common-overlap no-disturbance-without-uncertainty criterion of \cite{SunYu2019}, when closed under shared classical randomness, generates exactly the finite-dimensional quantum set whenever one party has two binary measurements. In particular, it recovers the full set of quantum behaviors in the simplest CHSH Bell scenario. 

The proof makes explicit how conditional probabilities satisfying the criterion can be represented by positive qubit operators on one side, while no-signalling ensures that the operators have a common average state on that side. A standard purification of that state then supplies the quantum realization. As a consequence, the quantum value of any linear Bell expression in these scenarios is given exactly by a one-parameter family of second-order cone programs, one for each overlap. The same exactness is available from steering considerations; the present formulation makes that reduction explicit as a single family of cones. The convexification of the NDWU criterion is essential; the common-overlap condition on its own is not preserved under classical shared randomization. When both parties have measurements with more outcomes, binary coarse-grainings of the behavior provide necessary outer conditions, but an intrinsic extension of the principle remains open for future work. 

\begin{acknowledgments}
We acknowledge support from the General Research Fund (GRF) Grant No.\ 17307925, and the Research Impact Fund (RIF) Grant No.\ R7035-21. Language-model tools were used during the preparation of this manuscript for discussion of the arguments, assistance with algebraic checks and editing. The theorems and proofs are the author's and the author is responsible for their correctness. 
\end{acknowledgments}

\onecolumngrid
\appendix
\newpage

\section{Quantitative version of the NDWU principle}
\label{app:ndwu}

Consider the Bell scenario in which two parties, Alice and Bob, each perform two measurements with two outcomes. Alice's input is \(x\in\{0,1\}\) and her outcome is \(a\in\{+1,-1\}\), analogously Bob's input is \(y\in\{0,1\}\) and his outcome is \(b\in\{+1,-1\}\), giving rise to a behavior \(\{P(a,b\mid x,y)\}\). Alice's input and outcome pair $(x,a)$ prepare Bob's system state. The conditional average of Bob's measurement \(B_y\) is then given by
\[
\mathbb E[B_y\mid A_x=a]
=
\sum_{b=\pm 1} b\,P(b\mid a,x,y)
=
\frac{\sum_{b=\pm 1} b\,P(a,b\mid x,y)}{P(a\mid x)},
\]
whenever \(P(a\mid x)>0\) where we used the no-signalling condition. 
Writing explicitly
$P(a,b\mid x,y)
=
\frac14\bigl(1+a\,\mathbb E[A_x]+b\,\mathbb E[B_y]+ab\,\mathbb E[A_x B_y]\bigr)$,
one has \(P(a\mid x)=(1+a\,\mathbb E[A_x])/2\) and
$\sum_{b=\pm 1} b\,P(a,b\mid x,y)
=
\frac{\mathbb E[B_y]+a\,\mathbb E[A_x B_y]}{2}$.
Dividing the two expressions gives
\begin{equation}
v_y^{a|x}
:=
\mathbb E[B_y\mid A_x=a]
=
\frac{\mathbb E[B_y]+a\,\mathbb E[A_x B_y]}{1+a\,\mathbb E[A_x]}.
\label{eq:cond-app}
\end{equation}
If \(P(a\mid x)=0\), the pair \((x,a)\) does not occur and imposes no constraint.

The NDWU principle is predicated on two assumptions that we restate here.
\begin{assumption}
\label{ass:sharp-app}
Both measurements are repeatable. If either measurement is performed twice in succession, the second outcome equals the first with probability one.
\end{assumption}

\begin{assumption}
\label{ass:reset-app}
The post-measurement states are outcome-determined: the state left by either measurement depends only on the outcome obtained. It does not depend on the preparation that entered the measurement.
\end{assumption}

The quantitative relation is imposed on two experiments. In the first, \(B_0\) is measured and then \(B_1\). In the second, \(B_1\) is measured and then \(B_0\). Each inequality is required for every preparation of the system, including the preparations left by the first measurement of either experiment. 
\begin{lemma}
\label{lem:opposite}
Let \(B_0,B_1\) satisfy Assumptions~\ref{ass:sharp-app} and~\ref{ass:reset-app}. Suppose that, for every preparation, the disturbance caused by measuring \(B_0\) before \(B_1\) denoted $D_{B_0 \to B_1}$ does not exceed the product of the two uncertainties, and that the same holds for the disturbance when \(B_1\) is measured before \(B_0\). Then there exists \(c\in[-1,1]\) such that
\[
\mathbb E[B_1\mid B_0=+1]
=
-\mathbb E[B_1\mid B_0=-1]
=
\mathbb E[B_0\mid B_1=+1]
=
-\mathbb E[B_0\mid B_1=-1]
=c.
\]
\end{lemma}
\begin{proof}
Denote \(c_+=\mathbb E[B_1\mid B_0=+1]\) and \(c_-=\mathbb E[B_1\mid B_0=-1]\), and denote \(c'_+=\mathbb E[B_0\mid B_1=+1]\) and \(c'_-=\mathbb E[B_0\mid B_1=-1]\), with no relation assumed among these four numbers. On a prepared state with averages \((v_0,v_1)\), measuring \(B_0\) first returns the average \((c_++c_-)/2+v_0(c_+-c_-)/2\) for the latter \(B_1\). The quantitative version of the NDWU relation then reads
\[
\left\lvert v_1-\frac{c_++c_-}{2}-v_0\frac{c_+-c_-}{2}\right\rvert
\le
\frac{\sqrt{1-c'_+{}^2}+\sqrt{1-c'_-{}^2}}{2}\,\sqrt{1-v_0^2}.
\]
The exchanged relation is the same with the two measurements interchanged. The state left by \(B_1=+1\) has averages \((c'_+,1)\), and the state left by \(B_1=-1\) has averages \((c'_-,-1)\). Substitute these two preparations, the sequential averages lie in \([-1,1]\), so that we obtain
\begin{align*}
1-\frac{c_++c_-}{2}-c'_+\frac{c_+-c_-}{2}
&\le
\frac{\sqrt{1-c'_+{}^2}+\sqrt{1-c'_-{}^2}}{2}\,\sqrt{1-c'_+{}^2},\\
1+\frac{c_++c_-}{2}+c'_-\frac{c_+-c_-}{2}
&\le
\frac{\sqrt{1-c'_+{}^2}+\sqrt{1-c'_-{}^2}}{2}\,\sqrt{1-c'_-{}^2}.
\end{align*}
Adding the two inequalities gives
\[
1-\frac{c_+-c_-}{2}\,\frac{c'_+-c'_-}{2}
\le
\left(\frac{\sqrt{1-c'_+{}^2}+\sqrt{1-c'_-{}^2}}{2}\right)^2.
\]
The graph of \(\sqrt{1-t^2}\) on \([-1,1]\) is the upper half of the unit circle, so the chord between any two points lies below the graph. By strict concavity of \(\sqrt{1-t^2}\) applied to the points \(c'_+\) and \(-c'_-\),
\[
\frac{\sqrt{1-c'_+{}^2}+\sqrt{1-c'_-{}^2}}{2}
\le
\sqrt{1-\left(\frac{c'_+-c'_-}{2}\right)^2},
\]
with equality only when \(c'_+=-c'_-\). Therefore
\[
\frac{c_+-c_-}{2}\,\frac{c'_+-c'_-}{2}
\ge
\left(\frac{c'_+-c'_-}{2}\right)^2.
\]
The same argument with the two measurements exchanged gives
\[
\frac{c_+-c_-}{2}\,\frac{c'_+-c'_-}{2}
\ge
\left(\frac{c_+-c_-}{2}\right)^2.
\]
Denote \(d=(c_+-c_-)/2\) and \(e=(c'_+-c'_-)/2\). Adding the two inequalities gives \(2de\ge d^2+e^2\), hence
\[
0\le(d-e)^2=d^2+e^2-2de\le 0,
\]
and therefore \(d=e\). Substituting this equality into the earlier bounding chain gives
\[
1-e^2
\le
\left(\frac{\sqrt{1-c'_+{}^2}+\sqrt{1-c'_-{}^2}}{2}\right)^2
\le
1-e^2.
\]
We obtain that equality holds, giving \(c'_+=-c'_-\). Analogously, the exchanged argument gives \(c_+=-c_-\). Finally \(d=e\) implies \(c_+=c'_+\), which is the claimed common parameter \(c\).
\end{proof}


Now consider a preparation on which \(B_y\) has average \(v_y\), so \(|v_y|\le 1\). Since the outcome is \(\pm 1\), one has \(B_y^2=1\), and the standard deviation is
$\sqrt{\mathbb E[B_y^2]-\mathbb E[B_y]^2}
=
\sqrt{1-v_y^2}$.
In particular, the standard deviation of \(B_0\) on the preparation is \(\Delta_{B_0}=\sqrt{1-v_0^2}\). By Lemma \ref{lem:opposite}, after the outcome $B_1 = +1$, the average of $B_0$ is $c$ and similarly after the outcome $B_1 = -1$ the average of $B_0$ is $-c$. Consequently
$\Delta_{B_0\mid B_1}
:= \frac12 \sum_{b = \pm 1} \sqrt{1 - \mathbb{E}[B_0 \mid B_1 = b]^2}
=
\sqrt{1-c^2}$.

Measuring \(B_1\) alone returns \(v_1\). Measuring \(B_0\) first returns \(+1\) with probability \((1+v_0)/2\) and \(-1\) with probability \((1-v_0)/2\). The subsequent averages of \(B_1\) are \(c\) and \(-c\), so the average of \(B_1\) at the end of the sequential experiment is
$\frac{1+v_0}{2}\,c+\frac{1-v_0}{2}\,(-c)=cv_0$. The disturbance is the absolute difference of the two averages,
$D_{B_0\to B_1}
:=
\lvert v_1-cv_0\rvert$.
No-disturbance-without-uncertainty is the statement that this disturbance cannot exceed the product of the two standard deviations~\cite{SunYu2019},
\begin{equation}
\sqrt{1-v_0^2}\,\sqrt{1-c^2}
\ge
\lvert v_1-cv_0\rvert.
\label{eq:principle-app}
\end{equation}
Both sides are nonnegative. Squaring and rearranging gives
\begin{equation}
v_0^2+v_1^2+c^2-2c\,v_0 v_1\le 1.
\label{eq:quad-app}
\end{equation}
This is the quadratic used in the main text. The same quadratic, with the parties exchanged, is the condition on Alice's side. We observe the following rewriting of the condition on the parameter $c \in [-1,1]$ to satisfy~\eqref{eq:quad-app}.

\begin{lemma}
\label{lem:interval}
Let \(v_0,v_1\in[-1,1]\). A number \(c\in[-1,1]\) satisfies~\eqref{eq:quad-app} if and only if \(c\) lies in the closed interval
\begin{equation}
I(v_0,v_1)
:=
\left[
v_0 v_1-\sqrt{(1-v_0^2)(1-v_1^2)},\;
v_0 v_1+\sqrt{(1-v_0^2)(1-v_1^2)}
\right].
\label{eq:interval-app}
\end{equation}
\end{lemma}
\begin{proof}
Rearranging~\eqref{eq:quad-app} as a quadratic in \(c\) gives \(c^2-2v_0 v_1\,c+v_0^2+v_1^2-1\le 0\). The discriminant is \(4(1-v_0^2)(1-v_1^2)\) and the roots are $c_{\pm} = v_0 v_1 \pm \sqrt{(1-v_0^2)(1-v_1^2)}$, which both lie in \([-1,1]\). Since the coefficient of $c^2$ is positive, the constraint \(c^2-2v_0 v_1\,c+v_0^2+v_1^2-1\le 0\) holds exactly between these roots in $[c_{-}, c_{+}]$, which is the interval $I(v_0,v_1)$.
\end{proof}

\subsection{Projective qubit measurements obey the NDWU principle}
\label{app:quantum}

Let \(B_0=\hat n_0\cdot\sigma\) and \(B_1=\hat n_1\cdot\sigma\), with unit vectors \(\hat n_0,\hat n_1\in\mathbb R^3\), and set \(c=\hat n_0\cdot\hat n_1\). A preparation with Bloch vector \(\vec r\), \(\lVert\vec r\rVert\le 1\), has averages \(v_y=\vec r\cdot\hat n_y\). The component of \(\vec r\) orthogonal to \(\operatorname{span}\{\hat n_0,\hat n_1\}\) does not increase the length of the projection onto that plane, so
\begin{equation}
(v_1-c v_0)^2
\le
(1-c^2)(1-v_0^2),
\label{eq:proj-bound}
\end{equation}
which is~\eqref{eq:quad-app}. The number \(c=\hat n_0\cdot\hat n_1\) does not depend on \(\vec r\), so the same \(c\) works for every preparation. The outcome \(B_0=+1\) leaves the Bloch vector \(\hat n_0\), and \(B_0=-1\) leaves \(-\hat n_0\). The subsequent averages of \(B_1\) are \(+c\) and \(-c\). Thus every behavior realised by one pure two-qubit state and projective measurements on both parties lies in \(\mathcal{A}_Y\).

\subsection{The plain NDWU set is not convex}
\label{app:nonconvex}

We recall the definitions from the main text. A behavior belongs to \(\mathcal{A}_B\) if there exists one number \(c_B\in[-1,1]\), which may depend on the behavior but not on \((x,a)\), such that
\[
\bigl(v_0^{a|x}\bigr)^2+\bigl(v_1^{a|x}\bigr)^2+c_B^2-2c_B v_0^{a|x}v_1^{a|x}\le 1
\]
for every \((x,a)\) with \(P(a\mid x)>0\). It belongs to \(\mathcal{A}_A\) if there exists one number \(c_A\in[-1,1]\), which may depend on the behavior but not on \((y,b)\), such that the same quadratic holds for the averages \(u_x^{b|y}\) of Alice's measurements. The set \(\mathcal{A}_Y\) is the intersection \(\mathcal{A}_A\cap\mathcal{A}_B\).

\begin{lemma}
\label{lem:det}
Every local deterministic behavior belongs to \(\mathcal{A}_Y\).
\end{lemma}
\begin{proof}
Denote, for \(\alpha_x,\beta_y\in\{+1,-1\}\), the behavior \(P_{\alpha,\beta}(a,b\mid x,y)=\delta_{a,\alpha_x}\delta_{b,\beta_y}\). Every occurring conditional pair has \(v_0=\beta_0\) and \(v_1=\beta_1\), both equal to \(\pm 1\). Then \(1-v_0^2=1-v_1^2=0\), so the square-root term in \(I(v_0,v_1)\) vanishes and \(I(\beta_0,\beta_1)=\{\beta_0\beta_1\}\). Equivalently, the quadratic becomes \((c-\beta_0\beta_1)^2\le 0\), so the only solution is \(c_B=\beta_0\beta_1\). The same calculation on Alice's side gives \(c_A=\alpha_0\alpha_1\). Therefore \(P_{\alpha,\beta}\in\mathcal{A}_Y\).
\end{proof}

\begin{lemma}
\label{lem:nonconvex}
\(\mathcal{A}_B\) is not convex. In particular, there exists an explicit local behavior that lies outside \(\mathcal{A}_B\), and hence outside \(\mathcal{A}_Y\). The same holds for \(\mathcal{A}_A\).
\end{lemma}
\begin{proof}
Denote \(P^{(+)}(a,b\mid x,y)=\delta_{a,+1}\delta_{b,+1}\) and \(P^{(-)}(a,b\mid x,y)=\delta_{a,-1}\delta_{b,(-1)^y}\). Both are local deterministic, hence both lie in \(\mathcal{A}_Y\) by Lemma~\ref{lem:det}, with \(c_B=+1\) and \(c_B=-1\) respectively. Their convex mixture \(P=\tfrac12 P^{(+)}+\tfrac12 P^{(-)}\) is local. The nonzero probabilities are \(P(+1,+1\mid x,0)=P(-1,+1\mid x,0)=\tfrac12\) and \(P(+1,+1\mid x,1)=P(-1,-1\mid x,1)=\tfrac12\). The steered averages on Bob's side are \(v_0^{a=+1|x}=v_1^{a=+1|x}=1\) and \(v_0^{a=-1|x}=1\), \(v_1^{a=-1|x}=-1\). For \((v_0,v_1)=(1,1)\) the square-root term vanishes and \(I(1,1)=\{1\}\). For \((v_0,v_1)=(1,-1)\) it vanishes and \(I(1,-1)=\{-1\}\). These intervals are disjoint, so no single \(c_B\) exists. Thus \(P\notin\mathcal{A}_B\) and \(P\notin\mathcal{A}_Y\). Since \(P\) is a convex mixture of two behaviors in \(\mathcal{A}_B\), the set \(\mathcal{A}_B\) is not convex, and so is \(\mathcal{A}_Y\). The same construction with the parties exchanged shows that \(\mathcal{A}_A\) is not convex.
\end{proof}

\section{Proof that \(\operatorname{conv}(\mathcal{A}_A)=\operatorname{conv}(\mathcal{A}_B)=\operatorname{conv}(\mathcal{A}_Y)=\mathcal{Q}_{2222}\)}
\label{app:hierarchy}

We restate the main theorem and prove it in this Appendix.

\begin{theorem}
\label{thm:suff-app}
It holds that
\begin{equation}
\operatorname{conv}(\mathcal{A}_A)
=
\operatorname{conv}(\mathcal{A}_B)
=
\operatorname{conv}(\mathcal{A}_Y)
=
\mathcal{Q}_{2222}.
\end{equation}
Furthermore, every point of \(\mathcal{Q}_{2222}\) is a convex combination of at most eight points of \(\mathcal{Q}_{2\mathrm{Q}}^{\mathrm{sharp}}\), and therefore of \(\mathcal{A}_Y\).
\end{theorem}

\begin{proof}
The argument has two directions. Throughout, \(\vec\sigma=(\sigma_x,\sigma_y,\sigma_z)\) denotes the three Pauli matrices.

Consider, without loss of generality, a behavior \(\{P(a,b\mid x,y)\}\in\mathcal{A}_B\), labelled by the eight parameters \((\alpha_x)_{x\in\{0,1\}}\), \((\beta_y)_{y\in\{0,1\}}\) and \((E_{xy})_{x,y\in\{0,1\}}\), the local marginals and the correlators. By Definition~\ref{def:sets}, membership in $\mathcal{A}_B$ implies that there is a single number \(c_B\in[-1,1]\), which may depend on \(P\) but not on \((x,a)\), such that
\begin{equation}
\bigl(v_0^{a|x}\bigr)^2+\bigl(v_1^{a|x}\bigr)^2+c_B^2-2c_B v_0^{a|x}v_1^{a|x}\le 1
\label{eq:B-app}
\end{equation}
for every \((x,a)\) with \(P(a\mid x)>0\). Here 
\begin{equation}
v_y^{a|x}=(\beta_y+a E_{xy})/(1+a\alpha_x).
\label{eq:vyax-def}
\end{equation}
We first show that this implies that $\{P(a,b\mid x, y)\} \in \Qset$.
There are two cases.

Case (i): \(|c_B|<1\). Denote \(s_B=\sqrt{1-c_B^2}\). Let us take the two measurement directions for Bob as the unit vectors
\begin{equation}
\hat n_0=(0,0,1),
\qquad
\hat n_1=(s_B,0,c_B),
\end{equation}
and the corresponding observables as \(B_y=\hat n_y\cdot\vec\sigma\). These are unit vectors since \(s_B^2+c_B^2=1\), and we have that \(\hat n_0\cdot\hat n_1=c_B\). For every \((x,a)\) with \(P(a\mid x)>0\), let us denote
\begin{equation}
\vec r_{a|x}
:=
\left(
\frac{v_1^{a|x}-c_B v_0^{a|x}}{s_B},\;
0,\;
v_0^{a|x}
\right),
\label{eq:r-app}
\end{equation}
and set \(\vec r_{a|x}=0\) if \(P(a\mid x)=0\). For every $(x,a)$ with $P(a \mid x) > 0$, the two inner products are \(\vec r_{a|x}\cdot\hat n_0=v_0^{a|x}\) and
$\vec r_{a|x}\cdot\hat n_1
=
s_B\cdot\frac{v_1^{a|x}-c_B v_0^{a|x}}{s_B}+c_B v_0^{a|x}
=
v_1^{a|x}$.
The squared length is
\begin{equation}
\|\vec r_{a|x}\|^2
=
\frac{(v_1^{a|x}-c_B v_0^{a|x})^2}{1-c_B^2}+\bigl(v_0^{a|x}\bigr)^2
=
\frac{(v_0^{a|x})^2+(v_1^{a|x})^2-2c_B v_0^{a|x}v_1^{a|x}}{1-c_B^2}.
\end{equation}
Thus \(\|\vec r_{a|x}\|\le 1\) if and only if \eqref{eq:B-app} holds. The eigenvalues of \(\vec r\cdot\vec\sigma\) are \(\pm\|\vec r\|\), so the eigenvalues of \((I+\vec r\cdot\vec\sigma)/2\) are \((1\pm\|\vec r\|)/2\), notably both are nonnegative when \(\|\vec r\|\le 1\). Let us denote
\begin{equation}
\tau_{a|x}
:=
P(a\mid x)\,\frac{I+\vec r_{a|x}\cdot\vec\sigma}{2}
\label{Eq:tauax-def}
\end{equation}
when \(P(a\mid x)>0\), and set \(\tau_{a|x}=0\) when \(P(a\mid x)=0\). Then by the above non-negativity of eigenvalues we have that \(\tau_{a|x}\succeq 0\).
We now compute the trace against Bob's observable. For \(P(a\mid x)>0\) we have from the above considerations that
\begin{equation}
\operatorname{Tr}(\tau_{a|x}B_y)
=
P(a\mid x)\,
\operatorname{Tr}\left[
\frac{I+\vec r_{a|x}\cdot\vec\sigma}{2}
\,(\hat n_y\cdot\vec\sigma)
\right] =
P(a\mid x)\,\vec r_{a|x}\cdot\hat n_y
=
P(a\mid x)\,v_y^{a|x}.
\label{eq:tr-tau-by}
\end{equation}
If \(P(a\mid x)=0\), then \(\tau_{a|x}=0\) and both sides are zero.

We now compute sum of $\tau_{a|x}$ over Alice's outcome and check that this is independent of her setting. Indeed we have
\begin{equation}
\sum_a\tau_{a|x}
=
\frac12 I
+
\frac12\Bigl(\sum_a P(a\mid x)\,\vec r_{a|x}\Bigr)\cdot\vec\sigma,
\end{equation}
from \(\sum_a P(a\mid x)=1\). The third component of the vector $\Bigl(\sum_a P(a\mid x)\,\vec r_{a|x}\Bigr)$ is \(\sum_a P(a\mid x)\,v_0^{a|x}\). From \(P(a\mid x)=(1+a\alpha_x)/2\) and the fact that \(v_0^{a|x} =  (\beta_0+a E_{x0})/(1+a\alpha_x)\) from \eqref{eq:vyax-def}, we have
\begin{equation}
P(a\mid x)\,v_0^{a|x}
=
\frac{\beta_0+a E_{x0}}{2} 
\label{eq:pax-v0ax}
\end{equation}
so that summing over \(a=+1\) and \(a=-1\) gives \(\beta_0\). If $P(a\mid x)=0$, positivity forces $P(a,b\mid x,y)=0$ for both $b$, so $\beta_y+a E_{xy}=0$; hence the same identity holds. The same calculation with \(y=1\) gives \(\beta_1\), so that the first component is \((\beta_1-c_B\beta_0)/s_B\). Now neither \(\beta_0\) nor \(\beta_1\) depends on \(x\), because each is Bob's marginal and the behavior $\{P(a,b|x,y)\}$ is no-signalling. Hence
\begin{equation}
\rho_B
:=
\sum_a\tau_{a|x}
=
\frac12\left[
I
+
\frac{\beta_1-c_B\beta_0}{\sqrt{1-c_B^2}}\,\sigma_x
+
\beta_0\,\sigma_z
\right]
\label{eq:rhoB-app}
\end{equation}
is the same operator for both \(x=0\) and \(x=1\).

Case (ii): \(c_B=\pm 1\). In this case, inequality \eqref{eq:B-app} becomes \((v_0^{a|x}-c_B v_1^{a|x})^2\le 0\), so \(v_1^{a|x}=c_B v_0^{a|x}\) whenever \(P(a\mid x)>0\). Set \(B_0=\sigma_z\) and \(B_1=c_B\sigma_z\). For \(P(a\mid x)>0\) set
\begin{equation}
\tau_{a|x}
=
\frac12\left[
P(a\mid x)\,I+P(a\mid x)\,v_0^{a|x}\,\sigma_z
\right],
\end{equation}
and set \(\tau_{a|x}=0\) if \(P(a\mid x)=0\). Then \(\operatorname{Tr}(\tau_{a|x} B_1)=c_B P(a\mid x)v_0^{a|x}=P(a\mid x)v_1^{a|x}\), as required. The same summation as above gives
\begin{equation}
\rho_B
:=
\sum_a\tau_{a|x}
=
\frac12\left[I+\beta_0\sigma_z\right],
\end{equation}
again independent of \(x\). In both cases \(0\le\tau_{+1|x}\preceq\rho_B\) and \(0\le\tau_{-1|x}\preceq\rho_B\), since these two operators are positive semidefinite and sum to \(\rho_B\). Consequently $\operatorname{supp}\tau_{a|x}\subseteq\operatorname{supp}\rho_B$, so the inverse and subsequent trace calculations may be restricted to that support. At this point we have the operators \(\{\tau_{a|x}\}\) on Bob's qubit, each is positive semidefinite, and both settings of Alice give the same sum \(\rho_B\). The rest of the argument follows from the SHJW theorem, see \cite{Hughston1993, Schrodinger1936}, we provide the details here.

We now define the joint state \(|\Psi\rangle\) of the two parties and Alice's measurement operators \(M_{a|x}\). This is the standard purification construction of Schrödinger and Hughston, Jozsa and Wootters~\cite{Schrodinger1936,Hughston1993}. The operator \(\rho_B\) is a \(2\times 2\) trace one positive semidefinite matrix, denote its eigenvalues  as \(\lambda_k\) 
and its rank by \(r \in \{1,2\}\).
Let \(\{|k\rangle\}_{k=1}^{r}\) be a corresponding orthonormal basis of eigenvectors, i.e., \(\rho_B|k\rangle=\lambda_k|k\rangle\) with \(\lambda_k>0\). Let us construct
\begin{equation}
|\Psi\rangle
:=
\sum_{k=1}^{r}\sqrt{\lambda_k}\,|k\rangle_A\otimes|k\rangle_B.
\label{eq:Psi-def}
\end{equation}
On \(\operatorname{span}\{|1\rangle,\ldots,|r\rangle\}\), we can now define Alice's measurement operators by
\begin{equation}
M_{a|x}
=
\sum_{k,\ell=1}^{r}
\frac{
\langle\ell|\tau_{a|x}|k\rangle
}{
\sqrt{\lambda_k}\,\sqrt{\lambda_\ell}
}
\,|k\rangle\langle\ell|.
\label{eq:M-app}
\end{equation}
In the eigenbasis of \(\rho_B\), this is the transpose of \(\rho_B^{-1/2}\tau_{a|x}\rho_B^{-1/2}\). Since \(\tau_{a|x}\succeq 0\), the matrix of \(\rho_B^{-1/2}\tau_{a|x}\rho_B^{-1/2}\) in that basis is positive semidefinite, and the transpose of a positive semidefinite matrix has the same eigenvalues, so that \(M_{a|x}\succeq 0\). And to see that these operators sum to the identity on the support, it suffices to substitute \(\sum_a\tau_{a|x}=\rho_B\) into \eqref{eq:M-app}, the only nonzero matrix elements of \(\rho_B\) in this basis are \(\langle\ell|\rho_B|k\rangle=\lambda_k\delta_{k\ell}\), so that
\begin{equation}
\sum_a M_{a|x}
=
\sum_{k=1}^{r}|k\rangle\langle k|.
\label{eq:B16}
\end{equation}
If \(r=2\), then \(\rho_B\) is invertible on \(\mathbb C^2\) and \(\sum_a M_{a|x}=I\). If \(r=1\), then \(M_{+1|x}\) and \(M_{-1|x}\) are the two numbers \(P(+1\mid x)\) and \(P(-1\mid x)\), which sum to one. Embedding the support into \(\mathbb C^2\) and assigning the orthogonal complement to one outcome (so that the effects sum to the full identity) gives a realization on \(\mathbb C^2\otimes\mathbb C^2\) without changing the probabilities.

With $M_{a|x}$ as in \eqref{eq:M-app} and \(N_{b|y}=(I+b B_y)/2\), the joint probability is the expectation in \(|\Psi\rangle\) which we compute now. From \eqref{eq:Psi-def}, we first obtain that
$\langle\Psi|(M_{a|x}\otimes N_{b|y})|\Psi\rangle
=
\sum_{i,j=1}^{r}
\sqrt{\lambda_i}\,\sqrt{\lambda_j}\,
\langle i|M_{a|x}|j\rangle\,
\langle i|N_{b|y}|j\rangle$.
From the definition of \(M_{a|x}\) in \eqref{eq:M-app}, the corresponding matrix element is
$\langle i|M_{a|x}|j\rangle
=
\frac{\langle j|\tau_{a|x}|i\rangle}{\sqrt{\lambda_i}\,\sqrt{\lambda_j}}$,
and we obtain
\begin{equation}
\langle\Psi|(M_{a|x}\otimes N_{b|y})|\Psi\rangle
=
\sum_{i,j=1}^{r}
\langle j|\tau_{a|x}|i\rangle\,
\langle i|N_{b|y}|j\rangle
=
\operatorname{Tr}(\tau_{a|x} N_{b|y}).
\end{equation}

Since \(N_{b|y}=(I+b B_y)/2\) and \(\operatorname{Tr}(\tau_{a|x})=P(a\mid x)\) from \eqref{Eq:tauax-def}, while \(\operatorname{Tr}(\tau_{a|x} B_y)=P(a\mid x)\,v_y^{a|x}\) from \eqref{eq:tr-tau-by}, we have that
\begin{equation}
\operatorname{Tr}(\tau_{a|x} N_{b|y})
=
\frac12\operatorname{Tr}(\tau_{a|x})
+
\frac{b}{2}\operatorname{Tr}(\tau_{a|x} B_y)
=
P(a\mid x)\frac{1+b\,v_y^{a|x}}{2}.
\end{equation}
For \(P(a\mid x)>0\), the definition of the conditional average is \(v_y^{a|x}=(\beta_y+a E_{xy})/(1+a\alpha_x)\) from \eqref{eq:vyax-def}, and \(P(a\mid x)=(1+a\alpha_x)/2\) by definition of $\alpha_x = \mathbb{E}[A_x]$, so that
\begin{equation}
P(a\mid x)\,v_y^{a|x}
=
\frac{\beta_y+a E_{xy}}{2}.
\end{equation}
Therefore, we finally have that 
\begin{align}
\langle\Psi|(M_{a|x}\otimes N_{b|y})|\Psi\rangle
&=
\operatorname{Tr}(\tau_{a|x} N_{b|y})
=
P(a\mid x)\frac{1+b\,v_y^{a|x}}{2} \nonumber \\
&=
\frac12\,P(a\mid x)
+
\frac{b}{2}\,P(a\mid x)\,v_y^{a|x}
=
\frac{1+a\alpha_x}{4}
+
\frac{b(\beta_y+a E_{xy})}{4}
\nonumber\\
&=
\frac14\bigl(1+a\alpha_x+b\beta_y+ab E_{xy}\bigr)
=
P(a,b\mid x,y).
\end{align}
If \(P(a\mid x)=0\), then \(\tau_{a|x}=0\), so the trace is zero, and the joint probability is also zero. Thus \(P\in\mathcal{Q}_{2222}\).



The same construction with the parties exchanged gives \(\mathcal{A}_A\subseteq\mathcal{Q}_{2222}\), and therefore \(\mathcal{A}_Y\subseteq\mathcal{Q}_{2222}\). The set \(\mathcal{Q}_{2222}\) is convex. If \(P=\sum_\lambda q_\lambda P_\lambda\) with \(q_\lambda\ge 0\), \(\sum_\lambda q_\lambda=1\) and
\begin{equation}
P_\lambda(a,b\mid x,y)
=
\operatorname{Tr}\bigl[\rho_\lambda(M_{a|x}^{(\lambda)}\otimes N_{b|y}^{(\lambda)})\bigr],
\end{equation}
the block-diagonal state \(\bigoplus_\lambda q_\lambda\rho_\lambda\) and the block-diagonal measurement operators realise \(P\). Hence
\begin{equation}
\operatorname{conv}(\mathcal{A}_Y)
\subseteq
\operatorname{conv}(\mathcal{A}_A)
\subseteq
\mathcal{Q}_{2222},
\label{eq:incls-1}
\end{equation}
and likewise with \(\mathcal{A}_B\) in place of \(\mathcal{A}_A\).

For the opposite inclusion, start from a finite-dimensional quantum behavior in \(\mathcal{Q}_{2222}\),
\begin{equation}
P(a,b\mid x,y)
=
\operatorname{Tr}\bigl[\rho_{AB}(M_{a|x}\otimes N_{b|y})\bigr],
\end{equation}
with \(\rho_{AB}\) a state on \(\mathcal H_A\otimes\mathcal H_B\) and \(\{M_{a|x}\}\), \(\{N_{b|y}\}\) binary POVMs. For Bob's setting \(y\), the effect \(N_{+1|y}\) satisfies \(0\le N_{+1|y}\le I\). Both binary POVMs admit projective dilations on a common finite-dimensional enlarged space with the same ancilla preparation: Naimark dilation gives an ancilla space, a vector \(|0\rangle\), and a projector \(\Pi_{+1|y}\) on the enlarged space such that
\begin{equation}
\operatorname{Tr}\bigl[\rho_{AB}(M_{a|x}\otimes N_{b|y})\bigr]
=
\operatorname{Tr}\bigl[
(\rho_{AB}\otimes|0\rangle\langle 0|)
(M_{a|x}\otimes \Pi_{b|y})
\bigr],
\end{equation}
where \(\Pi_{-1|y}=I-\Pi_{+1|y}\). The same dilation is applied to Alice's two POVMs. After these dilations, \(M_{+1|x}\) and \(N_{+1|y}\) are projectors, and the enlarged state still reproduces \(P\). We take \(M_{+1|x}\) and \(N_{+1|y}\) to be projectors from hereon, and write \(\Pi_{+1|y}\) for Bob's projectors and \(\Pi_{+1|x}\) for Alice's.

Jordan's Lemma states that if \(P\) and \(Q\) are two orthogonal projections on a finite-dimensional Hilbert space, then that space is an orthogonal direct sum of subspaces of dimension one or two, each invariant under \(P\) and under \(Q\)~\cite{Jordan1875}. We simply apply this to \(\{\Pi_{+1|0},\Pi_{+1|1}\}\) on each side. Then there are orthonormal decompositions
$\mathcal H_A=\bigoplus_{\mu}\mathcal H_{A,\mu},
\qquad
\mathcal H_B=\bigoplus_{\nu}\mathcal H_{B,\nu}$,
with \(\dim\mathcal H_{A,\mu}\le 2\) and \(\dim\mathcal H_{B,\nu}\le 2\), such that each local projection leaves every block invariant. The product decomposition
$\mathcal H_A\otimes\mathcal H_B
=
\bigoplus_{\mu,\nu}\mathcal H_{A,\mu}\otimes\mathcal H_{B,\nu}$
therefore reduces both parties' measurements. Denote by \(\Pi_{\mu\nu}\) the projection onto \(\mathcal H_{A,\mu}\otimes\mathcal H_{B,\nu}\), and set \(q_{\mu\nu}=\operatorname{Tr}(\rho_{AB}\Pi_{\mu\nu})\). On every block with \(q_{\mu\nu}>0\), denote
$\rho_{\mu\nu}
=
q_{\mu\nu}^{-1}\Pi_{\mu\nu}\rho_{AB}\Pi_{\mu\nu}$.
Let \(P_{\mu\nu}(a,b\mid x,y)\) be the behavior obtained by restricting the four projectors to this block and evaluating them in \(\rho_{\mu\nu}\). Since the projectors leave the blocks invariant, the cross terms across different blocks do not contribute, and we have that
\begin{equation}
P(a,b\mid x,y)
=
\sum_{\mu,\nu}q_{\mu\nu}P_{\mu\nu}(a,b\mid x,y).
\end{equation}
A one-dimensional block is spanned by a single vector, and the restricted measurement returns one of the two values \(+1\) or \(-1\) with probability one. We may embed that in \(\mathbb C^2\) by sending the basis vector to \(|0\rangle\), and represent the deterministic observable by \(\sigma_z\) or by \(-\sigma_z\), as per the value it returns. Every two-dimensional irreducible block has two rank-one outcome projectors for each setting, so its observables have the form $\hat{n} \cdot \vec{\sigma}$. Reducible blocks are split into one-dimensional blocks. 

Decompose each block state into pure states, \(\rho_{\mu\nu}=\sum_s p_s|\psi_s\rangle\langle\psi_s|\), with \(p_s\ge 0\) and \(\sum_s p_s=1\). The Born rule is linear in the state, so that we obtain
\begin{equation}
P_{\mu\nu}(a,b\mid x,y)
=
\sum_s p_s\,
\langle\psi_s|
\frac{I+a A_x}{2}
\otimes
\frac{I+b B_y}{2}
|\psi_s\rangle.
\end{equation}
Thus the behavior of \(\rho_{\mu\nu}\) is the same convex combination of the behaviors of the pure states \(|\psi_s\rangle\). Each summand is realised by one pure state in \(\mathbb C^2\otimes\mathbb C^2\) and one projective observable per setting. Recall that we denote by \(\mathcal{Q}_{2\mathrm{Q}}^{\mathrm{sharp}}\) the set of all behaviors obtained in this way: there exist a pure state \(|\psi\rangle\in\mathbb C^2\otimes\mathbb C^2\) and unit vectors \(\hat m_x,\hat n_y\in\mathbb R^3\) such that
\begin{equation}
P(a,b\mid x,y)
=
\bigl\langle\psi\bigr|
\frac{I+a\hat m_x\cdot\vec\sigma}{2}
\otimes
\frac{I+b\hat n_y\cdot\vec\sigma}{2}
\bigl|\psi\bigr\rangle.
\end{equation}
Therefore from the decomposition above we see that
\begin{equation}
\mathcal{Q}_{2222}
=
\operatorname{conv}\bigl(\mathcal{Q}_{2\mathrm{Q}}^{\mathrm{sharp}}\bigr).
\end{equation}
This decomposition is consistent with Masanes' extremal characterization~\cite{Masanes2005}.

On any such realization, as we have seen before, the overlap \(c_B=\hat n_0\cdot\hat n_1\) is a property of the two observables and not of the conditional state. The conditional state of Bob after the outcome \(A_x=a\) is a qubit state, so that its Bloch vector has length at most one. Its two averages are the inner products of that Bloch vector with \(\hat n_0\) and \(\hat n_1\), and the calculation in Appendix~\ref{app:ndwu} implies \eqref{eq:B-app} at this \(c_B\). The same holds for Alice with \(c_A=\hat m_0\cdot\hat m_1\), so that we obtain \(\mathcal{Q}_{2\mathrm{Q}}^{\mathrm{sharp}}\subseteq\mathcal{A}_Y\). We thus have
\begin{equation}
\mathcal{Q}_{2222}
=
\operatorname{conv}\bigl(\mathcal{Q}_{2\mathrm{Q}}^{\mathrm{sharp}}\bigr)
\subseteq
\operatorname{conv}(\mathcal{A}_Y)
\subseteq
\mathcal{Q}_{2222}.
\label{eq:incls-2}
\end{equation}
Combining \eqref{eq:incls-1} and \eqref{eq:incls-2} gives the statement of equalities in the theorem.

Finally, it remains to count the number of summands. A pure state of \(\mathbb C^2\otimes\mathbb C^2\cong\mathbb C^4\) is a line through the origin in \(\mathbb C^4\). The set of such lines is the complex projective space \(\mathbb{CP}^3\). Each projective observable with fixed outcome labels is a point of the unit sphere \(S^2\subset\mathbb R^3\). The parameter space \(\mathbb{CP}^3\times(S^2)^4\) is a product of compact sets, hence compact. It is path-connected: any two points can be joined by a continuous path lying entirely in the space (a pure state can be rotated continuously into any other pure state, and each unit vector on \(S^2\) can be rotated continuously into any other unit vector). The Born rule is the continuous map from this space to the sixteen probabilities. The continuous image of a compact set is compact, and the continuous image of a path-connected set is path-connected. Thus \(\mathcal{Q}_{2\mathrm{Q}}^{\mathrm{sharp}}\) is compact and path-connected. The no-signalling behaviors in this scenario form an eight-dimensional affine set. Carath\'eodory's theorem says that a point in the convex hull of a set in an eight-dimensional affine space is a convex combination of at most nine points of that set~\cite{Rockafellar1970}. Fenchel's refinement of Carath\'eodory's theorem says that if the set is path-connected, one fewer point suffices, namely eight~\cite{Fenchel1929,Rockafellar1970}. Therefore every point of \(\mathcal{Q}_{2222}\) is a convex combination of at most eight points of \(\mathcal{Q}_{2\mathrm{Q}}^{\mathrm{sharp}}\), and hence of \(\mathcal{A}_Y\).
\end{proof}

\section{The pure two-qubit set \(\mathcal{Q}_{2\mathrm{Q}}^{\mathrm{sharp}}\) is strictly smaller than \(\mathcal{A}_Y\)}
\label{app:sharp}

Lemma~\ref{lem:nonconvex} already shows that a local behavior can lie outside \(\mathcal{A}_Y\), so that the convex hull in Theorem~\ref{thm:suff-app} cannot be omitted. In this Appendix, we consider the set \(\mathcal{Q}_{2\mathrm{Q}}^{\mathrm{sharp}}\) consisting of those behaviors that are realized by some pure state of \(\mathbb C^2\otimes\mathbb C^2\) and projective measurements on both parties, and show that this set is strictly contained in \(\mathcal{A}_Y\), thereby elucidating the connection between the sets investigated in this paper.

Recall that we denote by \(\mathcal{Q}_{2\mathrm{Q}}^{\mathrm{sharp}}\) the set of \(2222\) behaviors \(P\) for which there exist a pure state \(|\psi\rangle\in\mathbb C^2\otimes\mathbb C^2\) and unit vectors \(\hat m_x,\hat n_y\in\mathbb R^3\) such that
\begin{equation}
P(a,b\mid x,y)
=
\bigl\langle\psi\bigr|
\frac{I+a\,\hat m_x\cdot\vec\sigma}{2}
\otimes
\frac{I+b\,\hat n_y\cdot\vec\sigma}{2}
\bigl|\psi\bigr\rangle.
\end{equation}
Masanes' theorem~\cite{Masanes2005} states that every extreme point of \(\mathcal{Q}_{2222}\) belongs to this set.

We first observe that this set is a subset of $\QY$. 
\begin{lemma}
\label{lem:sharp-in}
\(\mathcal{Q}_{2\mathrm{Q}}^{\mathrm{sharp}}\subseteq\mathcal{A}_Y\).
\end{lemma}
\begin{proof}
Consider any behavior \(\{P(a,b|x,y)\}\in\mathcal{Q}_{2\mathrm{Q}}^{\mathrm{sharp}}\), realized by a pure state \(|\psi\rangle\) and unit vectors \(\hat m_x,\hat n_y\). Let \(c_B:=\hat n_0\cdot\hat n_1\). This number is fixed by the two measurements of Bob and does not depend on Alice's setting or outcome $(x,a)$. For \(P(a\mid x)>0\), the conditional state of Bob is
\[
\rho_{a|x}
=
\frac{\operatorname{Tr}_A\bigl[((I+a\hat m_x\cdot\vec\sigma)/2\otimes I)\,|\psi\rangle\langle\psi|\bigr]}{P(a\mid x)},
\]
which is a qubit state, so its Bloch vector \(\vec r_{a|x}\) has length at most one. The conditional averages are the inner products \(v_y^{a|x}=\vec r_{a|x}\cdot\hat n_y\). The length calculation of case (i) in Appendix~\ref{app:hierarchy} converts \(\|\vec r_{a|x}\|\le 1\) into the NDWU quadratic given as
\begin{equation}
\bigl(v_0^{a|x}\bigr)^2+\bigl(v_1^{a|x}\bigr)^2+c_B^2-2c_B v_0^{a|x}v_1^{a|x}\le 1
\label{eq:C-app-1}
\end{equation}
for this \(c_B\). The same argument with the parties exchanged produces \(c_A=\hat m_0\cdot\hat m_1\). Thus \(\{P(a,b|x,y)\}\in\mathcal{A}_Y\).
\end{proof}

\begin{lemma}
\label{lem:AY-strict}
\(\mathcal{A}_Y\subsetneq\mathcal{Q}_{2222}\).
\end{lemma}
\begin{proof}
Theorem~\ref{thm:suff-app} gives \(\mathcal{A}_Y\subseteq\mathcal{Q}_{2222}\). The local mixture of Lemma~\ref{lem:nonconvex} lies in \(\mathcal{Q}_{2222}\) and not in \(\mathcal{A}_B\), hence not in \(\mathcal{A}_Y\).
\end{proof}

We now show that \(\mathcal{Q}_{2\mathrm{Q}}^{\mathrm{sharp}}\) is a strict subset of $\QY$ by exhibiting points in the latter set not in the former. 
\begin{lemma}
\label{lem:Pr}
Consider an arbitrary number \(r\in(2/3,1/\sqrt{2})\). Let \(\{P^{(r)}(a,b\mid x,y)\}\) be the behavior with \(x,y\in\{0,1\}\) and \(a,b\in\{+1,-1\}\) determined by \(\mathbb E[A_0]=\mathbb E[B_0]=-r/\sqrt{2}\), \(\mathbb E[A_1]=\mathbb E[B_1]=0\), \(\mathbb E[A_1 B_1]=-r\), and \(\mathbb E[A_0 B_0]=\mathbb E[A_0 B_1]=\mathbb E[A_1 B_0]=0\). Then \(\{P^{(r)}(a,b\mid x,y)\} \in \mathcal{L}(2,2;2,2) \cap \mathcal{A}_Y\), but \(\{P^{(r)}(a,b\mid x,y)\}\notin\mathcal{Q}_{2\mathrm{Q}}^{\mathrm{sharp}}\).
\end{lemma}
\begin{proof}
We first show that the behavior belongs to the local set $\mathcal{L}(2,2;2,2)$. Let the hidden variable be $\lambda=(a_0,b_0,a_1,b_1)\in\{\pm1\}^4$, with
$q(\lambda)
=
p_{00}(a_0,b_0)\,p_{11}(a_1,b_1)$, where
$p_{00}(a_0,b_0)=\frac14\Bigl(1-a_0\frac{r}{\sqrt{2}}-b_0\frac{r}{\sqrt{2}}\Bigr)$, and
$p_{11}(a_1,b_1)
=
\frac14(1-a_1 b_1 r)$.
Both probabilities are nonnegative, since \(r/\sqrt{2}<1/2\), we have that each of \(a_0 r/\sqrt{2}\) and \(b_0 r/\sqrt{2}\) has absolute value less than \(1/2\), and \(r<1\). On input \(x\), Alice returns \(a_x\). On input \(y\), Bob returns \(b_y\). The first pair $p_{00}(a_0,b_0)$ gives \(\mathbb E[A_0]=\mathbb E[B_0]=-r/\sqrt{2}\) and \(\mathbb E[A_0 B_0]=0\). The second pair $p_{11}(a_1,b_1)$ gives \(\mathbb E[A_1]=\mathbb E[B_1]=0\) and \(\mathbb E[A_1 B_1]=-r\). Independence of the two pairs gives \(\mathbb E[A_0 B_1]=\mathbb E[A_1 B_0]=0\). \(P^{(r)}\) admits a joint probability distribution and is thus local.

We now show that the behavior \(P^{(r)}\) belongs to \(\mathcal{A}_Y\). By Definition~\ref{def:sets}, this means that there is one number \(c_B\in[-1,1]\) such that, for every \((x,a)\) with \(P^{(r)}(a\mid x)>0\), the two conditional averages
\[
v_y^{a|x}
=
\mathbb E[B_y\mid A_x=a]
=
\frac{\mathbb E[B_y]+a\,\mathbb E[A_x B_y]}{1+a\,\mathbb E[A_x]}
\]
satisfy the quadratic~\eqref{eq:C-app-1}, and likewise on Alice's side with one number \(c_A\). For \(P^{(r)}\) these averages are
\begin{align*}
\mathbb E[B_0\mid A_0=+1]
&=
\frac{-r/\sqrt{2}}{1-r/\sqrt{2}},
&
\mathbb E[B_1\mid A_0=+1]
&=0,
\\
\mathbb E[B_0\mid A_0=-1]
&=
\frac{-r/\sqrt{2}}{1+r/\sqrt{2}},
&
\mathbb E[B_1\mid A_0=-1]
&=0,
\\
\mathbb E[B_0\mid A_1=+1]
&=
-r/\sqrt{2},
&
\mathbb E[B_1\mid A_1=+1]
&=-r,
\\
\mathbb E[B_0\mid A_1=-1]
&=
-r/\sqrt{2},
&
\mathbb E[B_1\mid A_1=-1]
&=r.
\end{align*}
We take \(c_B=0\), the quadratic ~\eqref{eq:C-app-1} then reduces to the sum of the squares of the two conditional averages $\bigl(v_0^{a|x}\bigr)^2+\bigl(v_1^{a|x}\bigr)^2$. For the first two lines that sum is the square of a number of absolute value at most \(1\), because \(r/\sqrt{2}<1/2\). For the last two lines it is \(r^2/2+r^2=3r^2/2<3/4<1\). The four conditional averages on Alice's side are the same four pairs, so \(c_A=0\) works as well. Thus \(P^{(r)}\in\mathcal{A}_Y\).

We now show that the behavior \(P^{(r)}\) is not in \(\mathcal{Q}_{2\mathrm{Q}}^{\mathrm{sharp}}\). Suppose, to the contrary, that there exists a pure two-qubit state and projective measurements \(A_x=\hat m_x\cdot\vec\sigma\), \(B_y=\hat n_y\cdot\vec\sigma\) realising \(P^{(r)}\). The number \(c_B=\hat n_0\cdot\hat n_1\) would then have to satisfy ~\eqref{eq:C-app-1} for every conditional pair above. At \(c_B=\pm 1\) the quadratic in ~\eqref{eq:C-app-1}  becomes \((\mathbb E[B_0\mid A_x=a]-c_B\mathbb E[B_1\mid A_x=a])^2\le 0\), so the two conditional averages would have to be equal or opposite. After the outcome \(A_0=+1\), these values are \((-r/\sqrt{2})/(1-r/\sqrt{2})\) and \(0\), which are neither equal nor opposite in the given ranges. The same holds for \(c_A\). Thus \(c_A\) and \(c_B\) would have to lie in \((-1,1)\).

Now, without loss of generality we take
\[
\hat m_0=\hat n_0=(1,0,0),
\qquad
\hat m_1=(c_A,s_A,0),
\qquad
\hat n_1=(c_B,s_B,0),
\]
where \(s_A=\sqrt{1-c_A^2}\) and \(s_B=\sqrt{1-c_B^2}\). Any two-qubit state has Bloch vectors \(\vec a,\vec b\in\mathbb R^3\) and a \(3\times 3\) correlation matrix \(T\), with \(\mathbb E[A_x]=\vec a\cdot\hat m_x\), \(\mathbb E[B_y]=\vec b\cdot\hat n_y\) and \(\mathbb E[A_x B_y]=\hat m_x^{\mathsf T}T\hat n_y\). These eight equations fix the components of \(\vec a\), \(\vec b\) and \(T\) in the plane of the two measurements as
\[
\vec a
=
\Bigl(-\frac{r}{\sqrt{2}},\;
\frac{r c_A}{s_A\sqrt{2}},\;
a_3\Bigr),
\qquad
\vec b
=
\Bigl(-\frac{r}{\sqrt{2}},\;
\frac{r c_B}{s_B\sqrt{2}},\;
b_3\Bigr),
\]
and the in-plane block of \(T\) as
\[
\begin{pmatrix}
T_{11}&T_{12}\\
T_{21}&T_{22}
\end{pmatrix}
=
\begin{pmatrix}
0&0\\
0&-r/(s_A s_B)
\end{pmatrix}.
\]
The third components \(a_3,b_3\) and the remaining entries of \(T\) are not fixed by \(P^{(r)}\).

To achieve a contradiction, consider any two unit vectors \(\vec u,\vec v\in\mathbb R^3\), with effects given by the operators \((I+\vec u\cdot\vec\sigma)/2\) and \((I+\vec v\cdot\vec\sigma)/2\). Their joint outcome probability in the state must be nonnegative
\[
\operatorname{Tr}\left[
\rho\,
\frac{I+\vec u\cdot\vec\sigma}{2}
\otimes
\frac{I+\vec v\cdot\vec\sigma}{2}
\right]
=
\frac14\Bigl[1+\vec a\cdot\vec u+\vec b\cdot\vec v+\vec u^{\mathsf T}T\vec v\Bigr]
\ge 0.
\]
Taking both vectors in the measurement plane, \(\vec u=\vec v=(1,\varepsilon,0)/\sqrt{2}\), with \(\varepsilon=\pm 1\), this supposed probability term is 
$
\frac14\Bigl[1-r-\frac{r}{2 s_A s_B}
+\varepsilon\frac{r}{2}\left(\frac{c_A}{s_A}+\frac{c_B}{s_B}\right)\Bigr]$. Now for any given fixed $c_A, c_B$, we may choose $\varepsilon = \pm 1$ so that the term $\varepsilon\frac{r}{2}\left(\frac{c_A}{s_A}+\frac{c_B}{s_B}\right)$ is $\leq 0$. We then have that since \(s_A s_B\le 1\), the value of the supposed probability is at most $\frac14\Bigl[1-r-\frac{r}{2}\Bigr]=\frac14\Bigl[1-\frac{3r}{2} \Bigr]
<0$. The joint probability term would be negative, and we have reached a contradiction. This contradicts the supposed realization. Therefore \(P^{(r)}\notin\mathcal{Q}_{2\mathrm{Q}}^{\mathrm{sharp}}\). 
\end{proof}

Thus \(\mathcal{Q}_{2\mathrm{Q}}^{\mathrm{sharp}}\subsetneq\mathcal{A}_Y\). The inclusion is Lemma~\ref{lem:sharp-in}, and the strictness is Lemma~\ref{lem:Pr}. The reverse strict inclusion \(\mathcal{A}_Y\subsetneq\mathcal{Q}_{2222}\) is Lemma~\ref{lem:AY-strict}. Combining the three lemmas, we obtain
\[
\mathcal{Q}_{2\mathrm{Q}}^{\mathrm{sharp}}
\subsetneq
\mathcal{A}_Y
\subsetneq
\mathcal{Q}_{2222}
=
\operatorname{conv}(\mathcal{A}_Y).
\]

\section{Beyond the \texorpdfstring{$(2,2;2,2)$}{(2,2;2,2)} scenario}
\label{app:beyond}

We prove two statements in this Appendix.
First, we show that the NDWU criterion generates the quantum set even in the Bell scenario where one party has two binary measurements and the other party has arbitrary finite numbers of settings and outcomes.
Second, binary coarse-grainings give a necessary condition when both parties have measurements with more outcomes. In going beyond the binary input-output scenario, one needs to be careful \cite{Slofstra2019, Pauwels2026, GJL26} so that throughout this Appendix, a quantum behavior means one admitting a finite-dimensional tensor-product realization.

\subsection{One party with two binary measurements}
\label{app:D-one-sided}

Consider the Bell scenario in which Alice has settings $x\in\{0,1\}$ and outcomes $a\in\{+1,-1\}$, while 
Bob has settings $y\in\{1,\ldots,n\}$ and outcomes $b\in\{1,\ldots,m_y\}$, where $n$ and every $m_y$ are positive integers.
Let us denote $\mathbf m=(m_1,\ldots,m_n)$.
A no-signalling behavior in this scenario is a family $P(a,b\mid x,y)\ge 0$ satisfying $\sum_{a,b}P(a,b\mid x,y)=1$, such that $\sum_a P(a,b\mid x,y)=P(b\mid y)$ is independent of $x$, and $\sum_b P(a,b\mid x,y)=P(a\mid x)$ is independent of $y$. Let us denote the finite-dimensional quantum set by $\mathcal Q(2,2;n,\mathbf m)$.
Thus $P\in\mathcal Q(2,2;n,\mathbf m)$ if there are finite-dimensional Hilbert spaces $\mathcal H_A,\mathcal H_B$, a positive operator $\rho_{AB}$ of trace one on $\mathcal H_A\otimes\mathcal H_B$, and measurements $\{M_{a|x}\}_a$, $\{N_{b|y}\}_b$ such that
\[
P(a,b\mid x,y)
=
\operatorname{Tr}\!\left[\rho_{AB}(M_{a|x}\otimes N_{b|y})\right].
\]

For $P(b\mid y)>0$, define
\begin{equation}
u_x^{b|y}
=
\mathbb E[A_x\mid B_y=b]
=
\frac{\sum_{a=\pm1}aP(a,b\mid x,y)}{P(b\mid y)}.
\label{eq:D-conditional}
\end{equation}
Let $\mathcal A_A^{(n,\mathbf m)}$ consist of the no-signalling behaviors for which there is one $c\in[-1,1]$ satisfying
\begin{equation}
\bigl(u_0^{b|y}\bigr)^2+\bigl(u_1^{b|y}\bigr)^2+c^2-2c\,u_0^{b|y}u_1^{b|y}\le 1
\label{eq:D-common-overlap}
\end{equation}
for every $(b,y)$ with $P(b\mid y)>0$.
As before, the same $c$ must work for all these conditional pairs, while an outcome with $P(b\mid y)=0$ imposes no conditional constraint.

\begin{theorem}
\label{thm:22n2sm}
For arbitrary finite $n$ and $\mathbf m$, it holds that
\begin{equation}
\mathcal Q(2,2;n,\mathbf m)
=
\operatorname{conv}\!\left(\mathcal A_A^{(n,\mathbf m)}\right).
\label{eq:D-generation}
\end{equation}
\end{theorem}

\begin{proof}
We adapt the two directions of the proof in Appendix~\ref{app:hierarchy}.
First, consider any arbitrary $P\in\mathcal A_A^{(n,\mathbf m)}$ and fix an admissible $c$. Let us set $w_{b|y}=P(b\mid y)$ and $d_{x,b|y}=\sum_{a=\pm1}aP(a,b\mid x,y)$.
Positivity, normalization, and no-signalling directly give that
\begin{align}
P(a,b\mid x,y)
&=\frac{w_{b|y}+a\,d_{x,b|y}}{2},
&
|d_{x,b|y}|&\le w_{b|y},
\label{eq:D-statistics}\\
\sum_b w_{b|y}&=1,
&
\sum_b d_{x,b|y}&=\alpha_x,
\label{eq:D-sums}
\end{align}
where $\alpha_x=\sum_a aP(a\mid x)$ is independent of $y$.
The common-overlap condition \eqref{eq:D-common-overlap} is then equivalent to
\begin{equation}
d_{0,b|y}^2+d_{1,b|y}^2-2c\,d_{0,b|y}d_{1,b|y}
\le (1-c^2)w_{b|y}^2.
\label{eq:D-homogeneous}
\end{equation}
This also holds when $w_{b|y}=0$, because~\eqref{eq:D-statistics} then forces both $d$'s to vanish.

For $|c|<1$, let us set $s=\sqrt{1-c^2}$ and define
\begin{align}
A_0&=\sigma_z,
&
A_1&=s\sigma_x+c\sigma_z,
\nonumber\\
\tau_{b|y}
&=
\frac12\left[
w_{b|y}I+
\frac{d_{1,b|y}-c\,d_{0,b|y}}{s}\,\sigma_x
+d_{0,b|y}\sigma_z
\right].
\label{eq:D-lift}
\end{align}
By~\eqref{eq:D-homogeneous}, we have that $\tau_{b|y}\succeq 0$. Its trace and observable averages are as before $\operatorname{Tr}(\tau_{b|y})=w_{b|y}$ and $\operatorname{Tr}(\tau_{b|y}A_x)=d_{x,b|y}$.
Moreover,~\eqref{eq:D-sums} gives the common sum
\begin{equation}
\sum_b\tau_{b|y}
=
\rho_A
:=
\frac12\left[
I+\frac{\alpha_1-c\alpha_0}{s}\,\sigma_x
+\alpha_0\sigma_z
\right],
\label{eq:D-common-state}
\end{equation}
which is independent of $y$.
In the case when $c=\pm 1$,~\eqref{eq:D-homogeneous} gives $d_{1,b|y}=c\,d_{0,b|y}$.
Let us choose in this case $A_0=\sigma_z$, $A_1=c\sigma_z$, and $\tau_{b|y}=\frac12(w_{b|y}I+d_{0,b|y}\sigma_z)$.
These operators are positive, have the same traces, and sum to $\rho_A=(I+\alpha_0\sigma_z)/2$.

Thus in either case $\{\tau_{b|y}\}_b$ is, for each $y$, a positive decomposition of the same qubit state.
As before, we apply the purification construction of Appendix~\ref{app:hierarchy} with the parties exchanged. That is, 
we take
\[
|\Psi\rangle
=
\sum_{k=1}^{r}\sqrt{\lambda_k}\,|k\rangle_A|k\rangle_B,
\qquad
N_{b|y}
=
\left(\rho_A^{-1/2}\tau_{b|y}\rho_A^{-1/2}\right)^{\mathsf T},
\]
where $\rho_A=\sum_{k=1}^{r}\lambda_k|k\rangle\langle k|$, $\lambda_k>0$, and $r\le 2$.
The inverse and transpose are taken on the support of $\rho_A$ in the displayed eigenbasis.
The operators $N_{b|y}$ are positive and sum to $I_r$.
The trace identity of Appendix~\ref{app:hierarchy} gives, with $M_{a|x}=(I+aA_x)/2$,
\[
\langle\Psi|M_{a|x}\otimes N_{b|y}|\Psi\rangle
=
\operatorname{Tr}(\tau_{b|y}M_{a|x})
=
\frac{w_{b|y}+a\,d_{x,b|y}}{2}
=
P(a,b\mid x,y).
\]
In the case when the rank $r=1$, we embed Bob's support into a qubit and add the projector onto its orthogonal complement to one outcome of each measurement.
This completes each measurement without changing the probabilities in the behavior.
Therefore $\mathcal A_A^{(n,\mathbf m)}\subseteq\mathcal Q(2,2;n,\mathbf m)$.
The direct-sum realization of finite mixtures, as in Appendix~\ref{app:hierarchy}, gives the same inclusion for the convex hull.

To get the inclusion in the opposite direction, consider any given finite-dimensional quantum realization of $P$.
As in Appendix~\ref{app:hierarchy}, we dilate Alice's two binary measurements to projective measurements on one enlarged space and apply Jordan's Lemma on Alice's side alone. The behavior then becomes a finite convex combination of behaviors whose spaces on Alice's side have dimension at most two. Decomposing each block state as before reduces the argument to pure states $|\psi\rangle\in\mathcal H_{A,j}\otimes\mathcal H_B$, with $\dim\mathcal H_{A,j}\le 2$.

The one additional step in this proof is to compress Bob's measurements. We write the Schmidt decomposition of one component as
\[
|\psi\rangle
=
\sum_{k=1}^{r}\sqrt{\mu_k}\,|a_k\rangle|\beta_k\rangle,
\qquad r\le 2,
\]
and let $W:\mathbb C^r\to\mathcal H_B$ be the isometry $W|k\rangle=|\beta_k\rangle$.
For every setting and outcome, we define $\widetilde N_{b|y} :=W^\dagger N_{b|y}W$. We check that these effects are positive and satisfy $\sum_b\widetilde N_{b|y}=I_r$.
With $|\widetilde\psi\rangle=\sum_k\sqrt{\mu_k}|a_k\rangle|k\rangle$, the identity $|\psi\rangle=(I\otimes W)|\widetilde\psi\rangle$ gives
\[
\langle\psi|F\otimes N_{b|y}|\psi\rangle
=
\langle\widetilde\psi|F\otimes\widetilde N_{b|y}|\widetilde\psi\rangle
\]
for every Alice operator $F$.
On a two-dimensional block, Alice's observables have the form $A_x=\hat m_x\cdot\vec\sigma$ with $\|\hat m_x\|=1$.
A one-dimensional block with deterministic outcomes $\eta_x$ is embedded into the qubit state $|0\rangle$ with observables $A_x=\eta_x\sigma_z$. Similarly Bob's one-dimensional supports are extended to qubits, so that we obtain the pure two-qubit realization.

Finally, we see that for every positive-probability outcome on Bob's side, Alice's conditional state is a qubit state. The qubit calculation in Appendix~\ref{app:ndwu} therefore gives~\eqref{eq:D-common-overlap} with the same overlap parameter $c=\hat m_0\cdot\hat m_1$ for all $(b,y)$.
Every component therefore belongs to $\mathcal A_A^{(n,\mathbf m)}$, proving the reverse inclusion and hence~\eqref{eq:D-generation}.
\end{proof}

For fixed $c\in[-1,1]$, set $s=\sqrt{1-c^2}$. The NDWU constraints are equivalently
\[
\left\|
\begin{pmatrix}
d_{1,b|y}-c\,d_{0,b|y}\\
s\,d_{0,b|y}
\end{pmatrix}
\right\|_2
\le s\,w_{b|y},
\]
together with the linear normalization and no-signalling constraints and $|d_{x,b|y}|\le w_{b|y}$. As a consequence, the quantum maximum of any linear Bell expression is again given exactly by a one-parameter family of second-order cone programs, one for each overlap; at $c=\pm 1$, these reduce to linear programs.

\subsection{Binary coarse-grainings of measurements with more outcomes}
\label{app:D-binning}

Consider the Bell scenario $(2,m;2,m)$ in which both parties have two measurements with $m\ge 2$ outcomes each. The inputs for Alice and Bob are $x,y\in\{0,1\}$, and the outputs are $a,b\in\{1,\ldots,m\}$ respectively.
For each setting, choose nonempty proper subsets $S_x\subsetneq\{1,\ldots,m\}$ and $T_y\subsetneq\{1,\ldots,m\}$.
Define
\[
f_x(a)=
\begin{cases}
+1,&a\in S_x,\\
-1,&a\notin S_x,
\end{cases}
\qquad
g_y(b)=
\begin{cases}
+1,&b\in T_y,\\
-1,&b\notin T_y.
\end{cases}
\]
We construct a coarse-grained behavior as
\begin{equation}
P_{S,T}(\alpha,\beta\mid x,y)
=
\sum_{\substack{a:f_x(a)=\alpha\\ b:g_y(b)=\beta}}
P(a,b\mid x,y),
\qquad
\alpha,\beta\in\{\pm1\}.
\label{eq:D-binned-behavior}
\end{equation}
Let $\mathcal R_m^{\mathrm{bin}}$ be the set of no-signalling behaviors in the $(2,m;2,m)$ scenario such that $P_{S,T}\in\operatorname{conv}(\mathcal A_Y)$ for every choice of the four subsets.
Here $\mathcal A_Y$ is the binary common-overlap set from the $(2,2;2,2)$ Bell scenario. We then have that

\begin{theorem}
\label{thm:2m2m}
For every finite $m\ge 2$,
\begin{equation}
\mathcal Q(2,m;2,m)
\subseteq
\mathcal R_m^{\mathrm{bin}}.
\label{eq:D-binning-inclusion}
\end{equation}
\end{theorem}

\begin{proof}
Take $P\in\mathcal Q(2,m;2,m)$, realized by $\rho_{AB}$ and measurements $\{M_{a|x}\}$, $\{N_{b|y}\}$.
Fix $S_0,S_1,T_0,T_1$, and define
\begin{align}
\widetilde M_{+1|x}
&=\sum_{a\in S_x}M_{a|x},
&
\widetilde M_{-1|x}
&=\sum_{a\notin S_x}M_{a|x},
\nonumber\\
\widetilde N_{+1|y}
&=\sum_{b\in T_y}N_{b|y},
&
\widetilde N_{-1|y}
&=\sum_{b\notin T_y}N_{b|y}.
\label{eq:D-binned-effects}
\end{align}
Every displayed operator is positive, and the two operators at each setting sum to the identity.
By linearity,
\begin{equation}
\operatorname{Tr}\!\left[\rho_{AB}(\widetilde M_{\alpha|x}\otimes\widetilde N_{\beta|y})\right]
=
P_{S,T}(\alpha,\beta\mid x,y).
\label{eq:D-binned-realization}
\end{equation}
Thus $P_{S,T}\in\mathcal Q_{2222}$.
Appendix~\ref{app:hierarchy} gives $\mathcal Q_{2222}=\operatorname{conv}(\mathcal A_Y)$, so $P_{S,T}\in\operatorname{conv}(\mathcal A_Y)$.
Since the subsets were arbitrary, $P\in\mathcal R_m^{\mathrm{bin}}$.


\end{proof}

More generally, for any finite Bell scenario, the same argument shows that every quantum behavior passes the convexified $(2,2;2,2)$ NDWU test after selecting two settings on each party and applying a binary binning at each setting.

\end{document}